\documentclass[11pt]{llncs}
\usepackage{amsmath,amssymb}
\usepackage{enumerate}
\usepackage{pict2e}
\usepackage{slashbox}
\usepackage[margin=1.05in]{geometry}
\usepackage{color}
\usepackage{xspace}
\usepackage[algoruled,vlined,english]{algorithm2e}
\usepackage{algorithmic}
\usepackage{placeins}
\usepackage{multirow}
\usepackage{tikz}
\usepackage[normalem]{ulem}
\newcommand{\sidecomment}[1]{~~\vspace*{-1em}\tcp{#1}}
\newcommand{\comments}[1]{{\color{blue}#1}}
\usepackage{environ}
\usepackage{setspace}
\begin{document}
\title{Design Patterns in Beeping Algorithms:\\ Examples, Emulation, and Analysis\thanks{A preliminary version of this paper was presented at OPODIS 2016. This research has been supported by ANR projects DESCARTES (ANR-16-CE40-0023) and ESTATE (ANR-16-CE25-0009-03).}}
\author{A. Casteigts, Y. M\'etivier, J.M. Robson and A. Zemmari\\
\institute{
Universit\'e de Bordeaux - 
Bordeaux INP\\
  LaBRI UMR CNRS 5800\\ 351 cours de la Lib\'eration, 33405 Talence\medskip\\
  \{acasteig, metivier, robson, zemmari\}@labri.fr 
}
} \date{ }

\maketitle

\begin{abstract}
We consider networks of processes which interact using beeps. In the basic model defined by Cornejo and Kuhn~(2010), processes can choose in each round either to beep or to listen. Those who beep are unable to detect simultaneous beeps. Those who listen can only distinguish between silence and the presence of at least one beep. We refer to this model as $BL$ (beep or listen). Stronger models exist where the nodes can detect collision while they are beeping ($B_{cd}L$), listening ($BL_{cd}$), or both ($B_{cd}L_{cd}$). Beeping models are weak in essence and even simple tasks are difficult or unfeasible within such models.
 
We present a set of generic building blocks ({\em design patterns}) which seem to occur frequently in the design of beeping algorithms.
They include {\em multi-slot phases}: the fact of dividing the main loop into a number of specialised slots; {\em exclusive beeps}: having a single node beep at a time in a neighbourhood (within one or two hops); {\em adaptive probability}: increasing or decreasing the probability of beeping to produce more exclusive beeps; {\em internal} (resp. {\em peripheral}) collision detection: for detecting collision while beeping (resp. listening); and {\em emulation} of collision detection when it is not available as a primitive.
Based on these patterns, we provide algorithms for a number of basic problems, including colouring, 2-hop colouring, degree computation, 2-hop MIS, and collision detection (in $BL$). The patterns make it possible to formulate these algorithms in a rather concise and elegant way. Their analyses are more technical;
one of them significantly reduces the constant factor in the analysis of the best known MIS algorithm by Jeavons et al.~(2016).
Finally, inspired by a technique from Afek et al.~(2013), our last contribution is to show that {\em any} Las Vegas algorithm relying on collision detection can be transposed into a Monte Carlo algorithm without collision detection, through emulation of this primitive at the cost of a logarithmic slowdown. We prove that this is optimal by showing a matching lower bound.

\end{abstract}

\section{Introduction}
\label{sec:intro}
The area of distributed computing is often concerned with the choice of assumptions. These assumptions may relate to the structure of the network (e.g. trees, rings, or complete graphs) or to the knowledge available to the nodes (network size, identifiers, port numbering). An important family of assumption is the size of the messages, which may range from unbounded to constant.
Clearly, given a problem, a natural goal is to reduce assumptions as much as possible. Thus, once a solution is found in some strong model, the community strives to solve it in weaker models. In a recent series of works~\cite{Cornejo10,Schneider10,Afek13,HuangM13,JSX16,GN15}, new models have been explored which are among the weakest imaginable. They are called {\em beeping models}. In these models, the only communication capability offered to the nodes is to {\em beep} or to {\em listen}. Several variants exist; in~\cite{Cornejo10}, a node that beeps is unable to detect whether other nodes have beeped simultaneously. When listening, it can distinguish between silence or the presence of at least one beep, but it cannot distinguish between one and several beeps. In Section~6 of~\cite{Afek13}, beeping nodes can detect whether other nodes are beeping simultaneously. In~\cite{Schneider10} and Section~4 of~\cite{Afek13}, yet another variant is considered where listening nodes can tell the difference between silence, one beep, and several beeps.

 In this paper, we denote the ability to detect collision while beeping (internal collision) by $B_{cd}$ and that of detecting collision while listening (peripheral collision) by $L_{cd}$. The absence of such ability is denoted by $B$ and $L$, respectively.
The existing models can be reformulated using the cartesian product of these capabilities. The basic model introduced by Cornejo and Kuhn in~\cite{Cornejo10} is $BL$; the model considered by Afek et al. in~\cite{Afek13} (Section~6) and Jeavons et al. in~\cite{JSX16} is $B_{cd}L$; and the model considered in~\cite{Schneider10} and in Section~4 of~\cite{Afek13} is $BL_{cd}$. In~\cite{CMRZ16c}, the authors consider the $B_{cd}L_{cd}$ model along with all others and their impact on the counting problem.
While some of these variants are stronger than others, all of them remain weak in essence. Beyond theory, these models account for real-world applications and phenomena. For instance, they reflect the features of a network at the lowest level (physical or MAC layer). They also reflect certain intercellular communication patterns in biological organisms~\cite{Collier96,Afek13,NB15}.

\subsection{Contributions}
Our first contribution is to identify generic building blocks
 ({\em design patterns}) which seem to occur often in the design of beeping algorithms.
 Based on these patterns, we present a number of algorithms for several graph problems, whose concise and simple expression illustrates the benefits of using design patterns. 
Then, we generalise a technique for using collision detection primitives when they are not available. This is done at the cost of a logarithmic slowdown, which we prove is optimal.
Another significant contribution is the complexity analysis of all the presented algorithms, as well as an improved analysis of the MIS algorithm from~\cite{JSX16}. These contributions are now presented in more detail.

\subsubsection{Design patterns.}
We identify a number of common building blocks in beeping algorithms, including {\em multi-slot phases}: the fact of dividing the main loop into a (typically constant) number of slots having specific roles ({\it e.g.,} contention among neighbours, collision detection, termination detection); {\em exclusive beeps}: the fact of having a single node beep at a time in a neighbourhood (within one or two hops, depending on the needs); {\em adaptive probability}: increasing or decreasing the probability of beeping in order to favor exclusive beeps; {\em internal} (resp. {\em peripheral}) collision detection: the fact of detecting collision while beeping (resp. listening); and {\em emulation} of collision detection: the fact of detecting collisions even when it is not available as a primitive. Relying on these patterns makes
 it possible to formulate algorithms in a concise and elegant way.

\subsubsection{Algorithms}
We present algorithms for a number of basic graph problems, including colouring, 2-hop colouring, degree computation, maximal independent set (MIS) and 2-hop MIS.
Quite often, the design of algorithms is easier and more
natural if collision detection is assumed as a primitive, e.g., in
$B_{cd}L_{cd}$ or $B_{cd}L$. Furthermore, emulation techniques like the ones described in this paper enable safe and automatic translations
of algorithms into weaker models such as $BL$. For this reason, our
algorithms are expressed using whatever model is the most
convenient and the analyses are formulated in the same model. All these algorithms are of type Las Vegas (i.e. guaranteed result, uncertain time), as opposed to Monte Carlo (i.e. guaranteed time, uncertain result).
First, we present a colouring algorithm in the
$B_{cd}L$ model which requires no knowledge about $G$ but may result in a large number of colours.  If the nodes know an upper bound $K$ on the maximum degree in the network $\Delta$, then a different strategy is proposed that uses only $K+1$ colours. A particular case is when $\Delta$ itself is known, resulting in $\Delta+1$ colours, which is a significant difference with the algorithm of Cornejo and Kuhn~\cite{Cornejo10} that results in $O(\Delta)$ colours with similar time complexity (both are not really comparable, however, since our algorithm is Las Vegas in $B_{cd}L$ and their algorithm is Monte Carlo in $BL$; their upper bound $K$ also is required to be within a constant factor of $\Delta$, which is not the case in our algorithm). Then, we show how to extend these algorithms to 2-hop versions of the same problems in the
$B_{cd}L_{cd}$ model, using the fact that a 2-hop colouring in a graph $G$ is equivalent to a colouring in the {\em square} of the graph $G^2$ (i.e. $G$ with extra edges for all paths of length two). Based on the observation that degree computation is strongly related to 2-hop colouring, we present a further adaptation of the algorithm to this problem. Finally, we turn our attention to the 2-hop MIS problem, which combines features from 2-hop colouring and from the MIS algorithm by Jeavons et al.~\cite{JSX16}. Similarities are observed and used for analysis.

\subsubsection{Analyses.}

One advantage of using design patterns is that it makes clear that the high-level behaviors of some algorithms are actually very general. For instance, the running time of the main colouring algorithm boils down to characterising the first moment when {\em every} node has produced an exclusive beep. What makes the characterisation of this time complexity difficult is the use of the {\em adaptive probability} pattern. Our first analysis shows that such a process takes $O(\Delta + \log n)$ time to complete with high probability\footnote{In this paper, with high probability ({\it w.h.p.}, for short) refers to probability $1-o(n^{-1})$.} The analysis
relies on a martingale technique with non-independent random
variables, which makes use of an old result by
Azuma~\cite{AK67}. We prove a matching $\Omega(\Delta + \log n)$ lower bound that establishes that our algorithm and analysis are essentially optimal. We show that the upper bound analysis applies to 2-hop colouring and degree computation almost directly, resulting in a $O(\Delta^2 + \log n)$ running time. Then, we present a new analysis of the MIS algorithm from~\cite{JSX16}. In terms of patterns, this algorithm also relies on {\em exclusive beeps} and {\em adaptive probability}, but it completes faster due to the fact that an exclusive beep by one node causes its {\em whole} neighbourhood to terminate. We prove that the complexity of
this algorithm is less than $76 \log n$ phases of two slots each (thus $152 \log n$ slots) {\it w.h.p.}, which improves dramatically upon the analysis presented in~\cite{JSX16}, which results in an upper bound of $e^{25}$ phases of two slots each.
Although constant factors are not the main focus in general, the gap here is one between practical and impractical running times. As such, our result confirms that the MIS algorithm by Jeavons et al. is practical. This analysis extends to the 2-hop MIS algorithm, resulting in $76 \log n$ phases of three slots each (thus $228 \log n$ slots). We also observe that the $\Omega(\log n)$ lower bound for colouring in the bit-size message passing model~\cite{KOSS} applies to the MIS problem in beeping models, making the algorithm from~\cite{JSX16} asymptotically optimal (up to a constant factor).
Finally, we characterise the running time of the ($K+1$)-colouring algorithm with known $K \ge \Delta$. This analysis is slightly less general because the high-level structure of the algorithm is strongly dependent on this particular problem. We show that $O(K \log n)$ slots are sufficient {\em w.h.p.}, which indeed corresponds to the same $O(\Delta \log n)$ as in~\cite{Cornejo10} if $K=O(\Delta)$. By a similar argument as above, this analysis extends to the 2-hop ($K^2+1$)-colouring in $O(K^2\log n)$ time. All complexities are summarised in Table~\ref{tab:results}.









\begin{table}
\small
\centering
\setstretch{1.2}
\begin{tabular}{|c|c|c|c|c|c|}
\hline

Problem &Article&Model &\#\,slots ({\it w.h.p.}) & Knowledge &  Comment \\

\hline
\hline
\multirow{3}{*}{Colouring} & this paper & $B_{cd}L$  & $\Theta(\Delta + \log n)$ & - & $O(\Delta + \log n)$ colours\\
 &Cornejo et al.~\cite{Cornejo10}& $BL$& ~~~~~$O(\Delta \log n)$ ~MC~& $K =O(\Delta)$&  $O(\Delta)$ colours\\
& this paper &  $B_{cd}L$  & $O(K\log n)$ & $K \ge \Delta$& $K+1$ colours\\
\hline

\multirow{2}{*}{$2$-hop colouring} & this paper &  $B_{cd}L_{cd}$  & $O(\Delta^2 + \log n)$ & - & $O(\Delta^2 + \log n)$ colours \\
 & this paper &  $B_{cd}L_{cd}$  & $O(\Delta^2 \log n)$ & $K=O(\Delta)$ & $K^2 +1$ colours \\

\hline
Degree computation & this paper &  $B_{cd}L_{cd}$  & $O(\Delta^2 + \log n)$ & - & -\\
\hline
  \multirow{2}{*}{MIS (\& 2-hop MIS)} & Jeavons et al.~\cite{JSX16} & $B_{cd}L$ &  $\le 2\,e^{25} \log n$ & - & \multirow{2}{*}{$\Omega(\log n)$\cite{KOSS}}\\
~& this paper & $B_{cd}L$ & $\le 152 \log n$&-& \\
\hline
  2-hop MIS & this paper & $B_{cd}L$ & $\le 228 \log n$&-& -\\
\hline

Emulation of $B_{cd}L_{cd}$
& \multirow{2}{*}{this paper} & \multirow{2}{*}{$BL$} & \multirow{2}{*}{$\Theta(\log n)$ factor} & \multirow{2}{*}{$N = O(n^c)$} & \multirow{2}{*}{LV $\to$ MC} \\
(or $B_{cd}L$ or $BL_{cd}$)& & & & &\\\hline
\end{tabular}\medskip
\setstretch{1}
\caption{\label{tab:results}Beeping algorithms on graphs with $n$ nodes, where $\Delta$ denotes the maximum degree in the graph; LV stands for Las Vegas; and MC stands for Monte Carlo. Unless otherwise mentioned, all algorithms are LV. }
\end{table}

\subsubsection{Collision detection and emulation techniques.}
Classical considerations on symmetry breaking in anonymous 
beeping networks, see for example
\cite{Afek13} (Lemma 4.1), imply that there is no Las Vegas internal collision
detection algorithm in the beeping models $BL$ and $BL_{cd}$. Likewise, there is no Las Vegas peripheral collision detection algorithm in the beeping models $BL$ and $B_{cd}L$. Since collision detection is required to detect exclusive beeps with certainty, and this pattern is central in most beeping algorithms, this implies that a large range of algorithms cannot exist in a Las Vegas version 
in these models.
Inspired by a technique used in Algorithm~3 of~\cite{Afek13}, we study the cost of detecting collision when it is not available and we present generic procedures to transform Las Vegas algorithms with collision detection into Monte Carlo algorithms in $BL$. We show that any collision in the neighbourhood of a {\em given} node
can be detected in  
$O(\log(\epsilon^{-1}))$ slots with error at most $\epsilon$, or in $O(\log n)$ slots {\it w.h.p.} Then, this is true for {\em all} nodes using $O(\log (\frac n \epsilon))$ slots with error $\epsilon$ or $O(\log n)$ slots {\it w.h.p.} We prove that this technique is asymptotically optimal (up to a constant factor) by giving a matching lower bound; {\it i.e.}, some topologies require $\Omega(\log n)$ slots to break symmetries {\it w.h.p.}



\subsection{Organisation of the paper} 

In Section~\ref{sec:definitions}, we present the model and give further definitions. Section~\ref{sec:patterns} introduces design patterns with several examples. The patterns are used in Section~\ref{sec:algorithms} to describe the various algorithms. For the sake of readability, the corresponding analyses are put together in a separate section (Section~\ref{sec:analysis}). Section~\ref{sec:emulation} presents and analyses the algorithms for collision detection in $BL$, and transposition techniques from higher models. 

\section{Network Model and Definitions}
\label{sec:definitions}

We consider a wireless network and
we follow definitions given by Afek. et al~\cite{Afek13} and Cornejo et al.~\cite{Cornejo10}.
The network is anonymous: unique identifiers are not available to
distinguish the processes. 
Possible communications are encoded by a graph $G=(V,E)$ where the nodes
$V$ represent processes and the edges $E$ represent pairs of processes
that can hear each other's beeps. 
We denote by $\Delta$ the maximum degree in $G$.
The neighbourhood of
a vertex $v$, denoted $N(v)$, is the set of vertices adjacent to $v$.
We define $\overline{N}(v)$ 
 by including $v$ itself in $N(v)$.
We also use the set of vertices at distance at most $2$ from $v$ called
the $2$-neighbourhood of $v$ and denoted $N_2(v)$ (or $\overline{N_2}(v)$ if it includes $v$).
Finally, we write $\log n$ for the binary logarithm of $n$ and $\ln n$ for the natural logarithm of $n$.

Time is divided into discrete synchronised time intervals (rounds) also called {\em slots}. 
All processes wake up and start computation in the same slot.
In each slot, all processors act in parallel and each can either beep or listen. 
In addition, processors can perform an unrestricted amount of local computation between two slots (this assumption is for compliance only, our algorithms do not require it).

\begin{remark}\label{listen}
In general, nodes are {\em active} or {\em passive}, depending on whether they are still taking part in the computation. When they are active
they beep or listen; in the description of algorithms we say explicitly when 
a node beeps meaning that a non beeping active node listens.
\end{remark}

The time complexity, also called {\em slot complexity}, is the maximum number of slots needed until all nodes have terminated. Our algorithms are typically structured into {\em phases}, each of which corresponds to a small (constant or logarithmic) number of slots. In the algorithm, we specify which one is the current slot by means of a {\tt switch} instruction with as many {\tt case} statements as there are slots in the phase. Phases repeat until some condition holds for termination.

\begin{remark}
An algorithm given in a beeping model induces an algorithm in the (synchronous)
message passing model. Thus, given a problem,
any lower bound on the round complexity
in the message passing model also holds for slot complexity in the beeping model.
\end{remark}

\paragraph*{Distributed Randomised Algorithm.}
A randomised (or probabilistic) algorithm is an algorithm which makes choices
 based on given probability distributions. 
A {\em distributed} randomised algorithm is a collection of local randomised
algorithms (in our case, all identical). 
A {\em Las Vegas} algorithm  is a randomised algorithm whose running time is not deterministic, but still finite with probability $1$, and that always produces
 a correct result.
A {\em Monte Carlo} algorithm is a randomised algorithm whose running time is deterministic, but whose result may be incorrect with a certain
probability. Put differently, Las Vegas algorithms have uncertain execution time but certain result, and Monte Carlo algorithms have certain execution time but uncertain result.
Classical considerations on symmetry breaking in anonymous 
beeping networks (see for instance Lemma~4.1 in~\cite{Afek13}), imply that:
\begin{remark}
There is no Las Vegas (and a fortiori no deterministic) algorithm in $BL$ which allows a node
to distinguish between an execution where it is isolated and one where it has exactly one neighbour.
\end{remark}

From this remark we deduce that there is no Las Vegas counting algorithm
in $BL$, which advocates the use of stronger models. In what follows, we consider whichever model is the most convenient and provide Las Vegas algorithms in these models. We then present canonical emulation techniques to turn any such algorithm into a Monte Carlo one in $BL$.

\paragraph{Graph problems.}

  Usually, the topology of a distributed system is modelled
by a graph and paradigms of distributed systems are represented by
classical problems in graph theory
such as computing the degree of the nodes,  computing a maximal independent set (MIS for short), a 
$2$-MIS  (i.e. a MIS in the square of $G$, that is, the same graph as $G$ with additional edges for distance 2 neighbors), a proper colouring
(a colouring  of a graph $G$ assigns colours to vertices
such that two neighbours have different colours), or a 
$2$-hop-colouring
(colouring of the square of $G$).
Each solution to one of these problems is a
building block for many distributed algorithms: symmetry breaking,
topology control, routing, resource allocation or network
synchronisation (see e.g.~\cite{Peleg,KMR01,EPSW14}).

\section{Design patterns for beeping algorithms}
\label{sec:patterns}

As a preliminary, this section presents a number of design patterns which seem to occur frequently in the design of beeping algorithms. The concept of pattern refers here to reusable solutions to common problems. These patterns are then used to describe algorithms in the other sections.

\paragraph*{Exclusive beeps.}
Beeping algorithms operate in synchronous periods called {\em slots}, which are equivalent to the concept of rounds in message passing models. Most problems in distributed computing require some node $v$ to take exclusive 
decisions at times (i.e., with respect to vertices of $\overline{N}(v)$ or $\overline{N_2}(v)$), which requires some type of symmetry breaking. In beeping networks, this goal is all the more difficult to achieve since the nodes cannot use identifiers nor even port numbers in their basic exchanges. If we assume that a node that is beeping can detect whether another node beeped simultaneously ($B_{cd}$), then this feature can be used to take an exclusive decision when indeed it beeped alone. We call this an {\em exclusive beep}. Algorithm~\ref{exclusive-beeps} illustrates an empty shell of algorithm that relies on repeated attempts to produce exclusive beeps. Most, if not all algorithms rely implicitly on this pattern as a basis.

\begin{algorithm}[h]
\caption{\label{exclusive-beeps}Exclusive beeps (using $B_{cd}$).}
\Repeat{$finished$}{\vspace{2pt}
        beep with some probability;\\
        \If{I beeped alone}{
          {\tt do something exclusive};
        }\vspace{2pt}
        ...\vspace{2pt}\\
}
\end{algorithm}

\paragraph*{2-hop exclusive beeps.}
For some problems like 2-hop colouring, 2-hop MIS, or computation of the degree (all discussed in this paper), the level of mutual exclusion offered by exclusive beeps is not sufficient and the algorithm requires that a node be the only one to beep within distance 2. Assuming collision can also be detected upon listening ($L_{cd}$), one can design a 2-slots pattern whereby non-beeping neighbours report if they have heard more than one beep. Hence, if a node produced an exclusive beep in the first slot, and none of its neighbours reported a collision in the second, then it knows that it has produced a {\em 2-hop exclusive beep} (Algorithm~\ref{2-hop-exclusive-beeps}).

\begin{algorithm}
\caption{\label{2-hop-exclusive-beeps}Two-hops exclusive beeps (using $B_{cd}L_{cd}$).}
\Repeat{$finished$}{
  \Switch{slot}{
\SetKwSwitch{Switch}{Case}{}{}{}{slot}{}{}{}%
    \uCase{1\sidecomment{contending}}{
        beep with some probability;
    }
    \uCase{2\sidecomment{detection of peripheral collision}}{
      \If{several neighbours beeped in slot 1}{beep}
    }
\SetKwSwitch{Switch}{Case}{}{}{}{after slot 2}{}{}{}%
    \uCase{}{ 
      \If{I beeped alone in slot 1 and no neighbour beeped in slot 2}
      {{\tt do something 2-hop exclusive}\\
      }
    }\vspace{3pt}
    ...\\
  }
}
\SetKwSwitch{Switch}{Case}{Other}{switch}{do}{case}{otherwise}{}{}
\end{algorithm}

\paragraph*{Multi-slot phases.}
Algorithm~\ref{2-hop-exclusive-beeps} illustrates another common aspect of beeping algorithms, namely {\em multi-slot phases}. The expressivity of a single beep is rather poor, but several combined slots can achieve elaborate behavior. In Algorithm~\ref{2-hop-exclusive-beeps}, one slot is devoted to contending and another to peripheral collision detection. The whole compound is then called a {\em phase}. 

\paragraph*{(Local) Termination detection.} In a multi-slot phase, one can add an extra so-called {\em termination slot}, in which all nodes that have not yet performed some action beep.
If a node's neighbours remain silent, then a form of local termination is detected, and the node can enter a terminal state (or switch to a subsequent activity).

\paragraph*{Adaptive probability.}
As far as feasibility and expressivity are concerned, the next design pattern is not crucial. However, it plays a central role in terms of performance. {\em Adaptive probability} consists in adapting the probability for a node to beep in the next phase depending on the outcome of previous phases. Typically, if a collision occurs, the probability is reduced, and if no one beeps, it is increased. Since the nodes do not know how many neighbours they are contending with (they do not know their degree), this technique increases significantly the odds of producing exclusive beeps.
\begin{algorithm}[h]
\caption{\label{adaptive}Adaptive beeping probability (using $B_{cd}L_{cd}$).}
$Float\ p \gets 1/2$ ~\tcp{say}
\Repeat{$finished$}{
        beep with probability $p$;\\
        \If{I beeped alone}{
          \tt do something exclusive;}
        \Else{
          \If{no one beeped}{
            increase $p$;
          }\Else{
            decrease $p$;
          }
        }
}
\end{algorithm}
Observe that the effective values are not given in Algorithm~\ref{adaptive}. Instead, we rely on the generic terms ``increase'' and ``decrease'' for generality of the pattern. In the analysis section,
 we use a doubling/halving pattern, that is, $p$ is increased to $2p$ (up to $1/2$), and it is decreased to $p/2$ (without limit).
 A similar doubling/halving pattern was used in~\cite{Scott13}. One could also increment or decrement the denominator of $p$ as done in~\cite{CMRZ16c}.
 The consequences of choosing one over the other are not discussed in this paper.

\paragraph*{Collision detection.}
Most algorithms in this paper use collision detection as a built-in primitive, referred to as $B_{cd}$ for detection on beeping and $L_{cd}$ for detection on listening. However, this feature is not always available as a primitive. An important question is the transformation of a (high-level) algorithm using $B_{cd}$ or $L_{cd}$ (or both) into one that works in the weakest $BL$ model. This question is the topic of Section~\ref{sec:emulation}, in which we study generic mechanisms to achieve this goal. Essentially, each slot that requires collision detection can be replaced with a logarithmic number of slots (in the size of various quantities depending on the desired guarantees) where the ties are broken
{\it w.h.p.} We provide dedicated procedures that generalise the technique used internally to one of the algorithms in~\cite{Afek13}. Besides complexity, the price to pay is that the algorithm becomes Monte Carlo instead of Las Vegas, that is, the result is correct only probabilistically, which is unavoidable. We present a matching lower bound that shows that these procedures are essentially optimal.

\section{Algorithms for basic graph problems}
\label{sec:algorithms}

We now present algorithms for a number of problems, including colouring (with or without information on the degree), 2-hop colouring, computation of the degree and 2-hop MIS. These algorithms are expressed using combinations of patterns presented in Section~\ref{sec:patterns}, which makes their exposition rather intuitive. We also recall Jeavons et al.'s Las Vegas algorithm for the MIS~\cite{JSX16}
problem and discuss its relations with our 2-hop MIS algorithm. All algorithms are Las Vegas and rely on whichever primitive ($B_{cd}L$ to $B_{cd}L_{cd}$ models) is convenient. The canonical adaptation of these algorithms in the weakest model ($BL$) is then discussed in Section~\ref{sec:emulation}.

\subsection{Colouring without knowledge}\label{sec:colouring}

The colouring problem consists of assigning a colour to every node in the network, such that no two neighbours have the same colour.
We first consider the case that no extra information is available to the nodes. 
Informally, the algorithm proceeds as follows (see Algorithm~\ref{algo:colouring} for details). Initially, every node is uncoloured ($nil$). In every phase, each node increments a counter. Uncoloured nodes contend with each other to produce an {\em exclusive beep}, and when one succeeds, it takes the current value of the counter as its colour and becomes passive. An {\em adaptive probability} is used to regulate the probability of beeping among active nodes. 
Local termination (a node and its neighbours are coloured)
detection is not explicitly handled here, although a {\em termination slot} could be added, where uncoloured nodes are the only ones to beep.

\begin{algorithm}[h]
\caption{\label{algo:colouring}Colouring algorithm in $B_{cd}L$ (without knowledge).}
$Float\ p \gets 1/2$;\\
$Integer\ colour \gets nil$;\\
$Integer\ counter \gets 0$;\\
\Repeat{$colour \ne nil$}{
        beep with probability $p$;\\
        \If{I beeped alone}{
          $colour \gets counter$}
        \Else{
          \If{no one beeped}{
            increase $p$;
          }\Else{
            decrease $p$;
          }
          $counter \gets counter + 1$;
        }
}
\end{algorithm}

We show in Section~\ref{sec:analysis} that 
the running time of this algorithm is of $O(\log n + \Delta)$ phases (in expectation and {\it w.h.p}), assuming a {\em doubling/halving} pattern is used for the adaptive probability.
 This also corresponds to the number of slots, since each phase consists of a constant number of slots. As to the number of colours, which is incremented in every phase, it is {\em at most} the same (some phases may not witness exclusive beeps).

\subsection{$(K+1)$-Colouring with a known bound $K \ge \Delta$}\label{sec:colouring-K}

If a bound $K\ge \Delta$ is known, then one can obtain a better colouring using at most $K+1$ colours. The algorithm follows the same lines as Algorithm~\ref{algo:colouring}, i.e. a colour counter is incremented in each phase, and its current value is chosen by those nodes who produce an exclusive beep. The main difference (see Algorithm~\ref{algo:colouring-K} for details) is that only those colours within $\{0,\dots, K\}$ are considered and thus the counter is incremented modulo $K+1$. Conflicts of colours are avoided by keeping a phase idle if the corresponding value was already taken in the past by a local neighbour. To do so, when a node takes a colour, it {\em re-beeps} in a dedicated {\em confirmation slot} to inform its neighbours that they must remove the current colour from their list of authorized colours. Accordingly, the uncoloured nodes will contend in a phase only if the current colour is still available (otherwise, they wait). An adaptive probability is used similarly to Algorithm~\ref{algo:colouring}, except that the probability is not updated in idle phases.

\begin{algorithm}[h]
\caption{\label{algo:colouring-K}$(K+1)$-Colouring algorithm in $B_{cd}L$ (knowing $K \ge \Delta$).}
$Colours$ $=\{0,\cdots,K\}$; \\
$Float\ p \gets 1/2$;\\
$Integer\ colour \gets nil;$ \\
$Integer\ counter \gets 0;$
 
\Repeat{$colour \ne nil$}{
  \If{$counter\in Colours$}{
    \Switch{slot}{
      \SetKwSwitch{Switch}{Case}{}{}{}{slot}{}{}{}%
      \uCase{1\sidecomment{contending}}{
        beep with probability $p$
      }
      \uCase{2\sidecomment{confirmation}}{
        \If{I beeped alone in slot 1}{ $colour \gets counter;$\\
          beep;
        }
        \Else{
          \If{no one beeped}{
            increase $p$;
          }\Else{
            decrease $p$;
          }
        }

      }
      \If{someone beeped in slot 2}
      {$Colours \gets Colours\setminus\{
        counter\}$} 
    }
  }
  $counter \gets (counter + 1) \mod (K+1)$;
}
\end{algorithm}

\subsubsection{A variant exploiting $K$ in the adaptive probability.}
The purpose of adaptive probability is to adjust the beeping probability to the local density (number of neighbors) and keep adjusting it as the execution progresses and the number of contenders decreases. If a bound $K$ on the degree is known, then a variant can be considered where the adaptive probability uses this information instead. Indeed, in the case of the $(K+1)$-colouring algorithm (Algorithm~\ref{algo:colouring-K}) the number of remaining colours is also a bound on the number of remaining contenders (and a good one if $K$ is close to the initial degree). We will consider such a variant in our analysis of Section~\ref{sec:analysis-colouring-K}, instead of the classical doubling/halving pattern (for which the dependencies proved difficult to manage). The exact variant we consider is as follows. We call a {\em cycle} a sequence of $K+1$ phases. In the beginning of each cycle, every node updates its beeping probability $p$, setting it to one over twice the number of unused colours. In this context, we prove that the number of phases is $O(K \log n)$ {\it w.h.p.}.
Intuitively it would seem more reasonable to always use the most recent information on the
number of available colours rather than the number at the beginning of the current cycle
but we can only prove a weaker result for such an algorithm.

\comments{
 

 
}

\subsection{2-hop colouring}\label{2col}
\label{sec:2-hop-colouring}

A 2-hop colouring of a graph $G$ is a colouring such that any two nodes at distance $\le 2$ have different colours. In other words, it is a colouring of the square of $G$, the graph where an edge exists between nodes which are neighbours in $G$ or share a common neighbour in $G$.

\paragraph*{2-hop colouring without knowledge.}

A similar strategy is used as in Algorithm~\ref{algo:colouring} (colouring), except that exclusive beeps are replaced with {\em 2-hop exclusive beeps}. Whenever a node produces such a beep, it takes the current value of the counter as colour. Since no other node has beeped within distance $2$, the colouring is legal. Contrary to the 1-hop colouring, the collaboration of a node remains crucial even after it becomes coloured. Indeed, this node must keep on reporting peripheral collisions to its neighbours. As a result, instead of retiring from computation, coloured nodes keep on listening until all of their neighbours are coloured, which is detected using an extra {\em termination slot}. Details are given in Algorithm~\ref{algo:2-hop-colouring}. Four slots are used in total, the first two being devoted to the management of 2-hop exclusive beeps (see Section~\ref{sec:patterns} for details). The third slot manages a (2-hop) adaptive probability based on beeps heard at distance one (slot 1) or at distance two (slot 3 itself). Finally, slot 4 is the termination slot.
\begin{algorithm}[h]
\caption{\label{algo:2-hop-colouring}2-hop colouring algorithm in $B_{cd}L_{cd}$ (without 
knowledge).}
$Float\ p \gets 1/2$;\\
$Integer\ colour \gets nil$;\\
$Integer\ counter \gets 0$;

\Repeat{no beep heard in slot $4$}{
  \Switch{slot}{
    \SetKwSwitch{Switch}{Case}{}{}{}{slot}{}{}{}%
    \uCase{1\sidecomment{contending slot}}{
      \If{$colour = nil$}{
        beep with probability $p$;
      }
    }
    \uCase{2\sidecomment{peripheral collision detection (and consequences)}}{
      \If{several neighbours beeped in slot 1}{beep}
      \If{I beeped alone in slot 1 and heard no beep in slot 2}
      {$colour \gets counter$}
    }
    \uCase{3\sidecomment{adaptive probability}}{ 
      \If{someone beeped in slot 1}{beep}
      \If{$colour = nil$}
      { 
        \If {no beep heard in slot 1 nor 3}
        {increase $p$}
        \Else {decrease $p$}}
    }

    \uCase{4\sidecomment{termination slot}}{ 
      \If{$colour = nil$}{beep} 
    }
  }
  $counter \gets counter + 1$
}

\end{algorithm}

Once we realize that the execution produced here is the same as what Algorithm~\ref{algo:colouring} would produce in the square of $G$, analysis of this algorithm is straightforward. The only difference is that the maximal number of contenders of a node becomes $\Delta^2$ instead of $\Delta$. Thus Algorithm~\ref{algo:2-hop-colouring} takes $O(\log n + \Delta^2)$ phases (and slots) {\it w.h.p.}, and the number of colours cannot exceed the same value.

\paragraph*{With a bound $K$ on the maximum degree $\Delta$.}
The same idea can be applied as in the 1-hop variant, {\it i.e.}, taking colours between $0$ and $K^2 $ (instead of $K$) and incrementing the counter modulo $K^2 + 1$. As a result, at most $K^2+1$ colours are used, with time complexity $O(K^2\log n)$ {\it w.h.p.}

\subsection{Degree computation}
\label{sec:degree}

In this paragraph, we discuss the degree computation problem, which consists, for every node, of computing the number of its neighbours in the network.
Let us recall that 2-hop exclusive beeps allow a node to perform an exclusive action within a radius of distance 2. This feature was used in Algorithm~\ref{algo:2-hop-colouring} to assign unique colours. It turns out that the pattern is quite versatile and can also be used, for instance, to compute the degree of a node. In more detail, a 2-hop exclusive beep realised by some node $v$ implies that all neighbours of $v$ become aware of its presence (they all increment their own degree), then $v$ stops contending and keeps on reporting peripheral collisions. The corresponding modifications of Algorithm~\ref{algo:2-hop-colouring} are straightforward. They consist of having a new confirmation slot inserted, in which $v$ re-beeps if indeed it produced a 2-hop exclusive beep. Upon hearing the confirmation beep, all of $v$'s neighbours increment a local counter that eventually amounts to their degree. Termination proceeds as before, i.e. all uncounted nodes beep in a termination slot so that local termination is detected.
Up to a change in the constant factor, which accounts for the additional confirmation slot in each phase, the running time of this algorithm remains unchanged, that is $O(\log n + \Delta^2)$ slots (or phases) {\it w.h.p}.

\subsection{MIS and 2-hop MIS}
\label{sec:mis}

The {\em maximal independent set} problem (MIS, for short) consists of selecting a set of nodes, none of which are pairwise neighbours, and such that the set is maximal for the inclusion relation. In~\cite{JSX16}, Jeavons et al. present a (Las Vegas) beeping algorithm for the MIS problem in the $B_{cd}L$ model. Thanks to the patterns introduced in Section~\ref{sec:patterns}, this algorithm can be described in a very concise and intuitive way as follows. Like most, this algorithm relies on having the nodes produce exclusive beeps in competition with each other. Whenever a node succeeds, it enters the MIS. So far, the process is quite similar to that of the colouring algorithm presented above. A fundamental difference is that once a node has produced such a beep, then its {\em whole} neighborhood terminates at once. Indeed, if a node enters the MIS, then all of its neighbours are settled at the same time about their membership (i.e. not in the MIS). This elimination is made through a confirmation slot in which a node re-beeps if it has produced an exclusive beep (as seen above). Finally, an adaptive probability (with doubling/halving pattern) is used to maximise the odds of producing exclusive beeps.

The elimination of all neighbours upon exclusive beeps makes the process faster to terminate. Jeavons et al. prove that it terminates within $O(\log n)$ phases, however with an upper bound of $e^{25}$ on the constant factor. Although a much lower complexity is observed in practice by the authors, they did not attempt to establish better bounds analytically. In Section~\ref{sec:analysis-MIS}, we present a new analysis of this algorithm that takes the constant down to $76$ (for the number of {\em phases}), thereby confirming the authors empirical evidence that their algorithm is practical. Although constant factors are not the main focus in complexity, the gap in this case is one between practical and impractical running times, which we believe makes our analysis a substantial contribution.

\paragraph{Computing a 2-hop MIS.}

A 2-hop MIS~\cite{MIS2} is a set of nodes such that no pair of selected nodes are within distance 2 and the set is maximal under ordering by inclusion. Similarly to the colouring algorithm presented above, the simple observation that a 2-hop MIS is a MIS in the square of the graph allows one to extend the (1-hop) algorithm to the 2-hop case fairly easily. Namely, whenever a node produces a 2-hop exclusive beep, it enters the MIS and informs its neighbours (using a confirmation slot) that they are eliminated. Unlike the colouring algorithms, there is no dependence here on $\Delta$ (the largest degree) in the time complexity. As a result, the number of phases of the 2-hop algorithm, which is equivalent to that of the 1-hop algorithm in the square of the graph, remains of $O(\log n)$ {\it w.h.p.} (with the same constant factor of $76$). As before, the number of {\em slots} is slightly scaled due to the additional confirmation slot in each phase.

\section{Collision detection and emulation techniques}
\label{sec:emulation}

\newcommand{\phase}{sub-phase\xspace}
\newcommand{\phases}{sub-phases\xspace}

In the previous sections, we have considered collision detection as an available primitive. Depending on the algorithms, we assumed that collision detection was possible while beeping ($B_{cd}$) or while listening ($L_{cd}$). This assumption is convenient because it allows one to design simple algorithms. Furthermore, it allows the algorithms to be of a {\em Las Vegas} type (see Section~\ref{sec:definitions} for definitions). Unfortunately, we know since~\cite{Afek13} that no Las Vegas algorithms can be designed for non-trivial problems without collision detection, that is, in the $BL$ model. One has to turn to Monte Carlo instead, which means that the result is correct only with some probability (possibly {\em w.h.p.}). In this section, we investigate the cost of building a probabilistic collision detection primitive in the $BL$ model, generalising a technique used inside one of the algorithms of~\cite{Afek13}. Then, we adapt this technique into two generic emulation procedures, one for detecting collision while beeping, the other while listening. These procedures can subsequently be used to translate any Las Vegas algorithm expressed in $B_{cd}L$, in $BL_{cd}$, or in $B_{cd}L_{cd}$, into a Monte Carlo algorithm in $BL$. The cost to pay is a logarithmic slowdown of the running time, which we prove is essentially optimal (for sufficiently large $n$) thanks to a matching lower bound. 

\subsection{Collision detection}
\label{sec:detection}

\comments{
}
The impossibility for a node in $BL$ to distinguish between being alone or having neighbours has strong implications. For instance, in the colouring problem, it means that two neighbours could possibly end up with the same colour. In the MIS problem, two neighbours could enter the MIS. In fact, there is no guarantee on the correctness of basic patterns like exclusive beeps or 2-hop exclusive beeps, which are at the basis of most (if not all) Las Vegas algorithms. 

We present a (Monte Carlo) algorithm for detecting collisions
in $BL$. This procedure generalises the technique used in Algorithm~3 of~\cite{Afek13}, which consists in replacing a slot that requires collision detection in the original model, by several $BL$ slots in which symmetries are {\em probabilistically} broken. Of course, the more slots are paid, the more reliable the detection. Later on, we investigate the tradeoff between different levels of guarantees and different numbers of $BL$ slots per original slot.

\paragraph*{Details of the algorithm.}
Each slot that requires collision detection ($B_{cd}$ or $L_{cd}$) is replaced by a number of {\em \phases}, each consisting of two $BL$ slots. For instance, if a node wishes to beep with collision detection in the original algorithm, it will choose one of the two slots ({\em u.a.r.}) in each of the \phases and will beep in that slot (listen in the other). If it hears a beep while listening in the other slot, then an internal collision is detected. Similarly, if a node wishes to listen with collision detection in the original algorithm, it will listen in both slots of each \phase. A peripheral collision is detected if a beep is heard in both slots of a same \phase. The procedure is detailed by Algorithm~\ref{algo:detection}, where $k$ is the number of \phases used.

\begin{algorithm}[h]
\caption{\label{algo:detection}Collision detection algorithm in $BL$ (with parameter $k$)}
$Boolean\ collision \gets false$;\\
$Integer\ i \gets 0;$\\
\While{$i < k$}{
  \If{$v$ wishes to beep}
  {Flip a coin; \\
    \If{{\tt heads}}
    {beep in slot 1;\\listen in slot 2;}
    \Else{listen in slot 1;\\beep in slot 2;}
    \If{another beep was heard}{$collision \gets true$}
  }
  \Else{listen in both slots;\\
    \If{beeps are heard in both slots}{$collision \gets true;$}
  }
  $i \gets i+1$;
}
{return $collision$;}
\end{algorithm}

False positives never happen, but real collisions might go unnoticed, with probability inversely related to the number $k$ of sub-phases. We are interested in determining how large $k$ should be to guarantee that a {\em given} node detects a collision in its neighbourhood with good probability. The stronger question asks how many \phases are required so that {\em none} of the nodes fails to detect a potential collision ({\it w.h.p.})

\begin{lemma}
\label{lemma::collision_local}
Let $v$ be a node. If a collision occurs in the neighbourhood of $v$, then 
$v$ detects it in $O\left(\log(\frac 1 \epsilon)\right)$ \phases (slots)
with
probability at least $1-\epsilon,$ and in $O\left(\log n\right)$ \phases
(slots) with probability $1-o\left(\frac 1 {n^2}\right)$.
\end{lemma}
\begin{proof}
Assume a collision occurs between some nodes $u_1$ and $u_2$ in the neighbourhood of $v$ (one of them being possibly $v$ itself). It is detected if $u_1$ and $u_2$ choose a different slot in at least one of the $k$ \phases. 
The probability that this does not happen is
$\left(\frac 1 {2}\right)^k$. This probability is less than
$\epsilon$ (resp. $o\left(\frac 1 {n^2}\right)$) for any
$k \ge \log(\frac 1 \epsilon)$
(resp.  $2\log(n)$). Observe that if collisions occur between more than two nodes in the neighbourhood of $v$, this cannot decrease the odds of a successful detection (to the contrary, the odds can only increase).\qed
\end{proof}

\begin{corollary}
Let $G$ be a graph. If collisions occur in the neighbourhood of an arbitrary number of nodes, then all of them detect collision after at most $O\left(\log(\frac
n \epsilon)\right)$ \phases (slots) with probability at least $1-\epsilon$, and
after at most $O\left(\log n\right)$ \phases (slots) {\it w.h.p.}
\end{corollary}
\begin{proof}
Assume collisions occur in $G$ and let $T$ denote the number
of \phases before all concerned nodes detect collision. Clearly $T=\max\{
T_v\mid v\in V\}$, where $T_v$ is the time it takes any node $v$ to decide collision. By the same argument as in the proof of Lemma~\ref{lemma::collision_local}, together with the union bound, it holds that

\begin{eqnarray}
{\mathbb P}r
\left(
T>\log\left(\frac n \epsilon\right)\right)
& \leq & n\times 
{\mathbb P}r
\left(
T_v>\log\left(\frac n \epsilon\right)\right)\\
& = & n\times \frac 1 {2^{\log(\frac n \epsilon)}}= \epsilon
\end{eqnarray}
which proves the first claim. The same argument, combined with the
second claim of Lemma \ref{lemma::collision_local} proves the second
claim.\qed
\end{proof}

\subsection{Emulation procedures}\label{sec:emuler}

Based on the tie-breaking mechanism presented in Algorithm~\ref{algo:detection},
 we define two probabilistic emulation procedures whose purpose is 
to replace beep or listen instructions with collision detection in $BL$. 
Both are Monte Carlo in the sense that detection is only guaranteed 
with some probability.
The first procedure, {\tt EmulateB$_{cd}$inBL()}, 
is given by Algorithm~\ref{algo:emulateBcd} and the second, 
{\tt EmulateL$_{cd}$inBL()}, by Algorithm~\ref{algo:emulateLcd}.
 Both procedures are
parametrized by an integer $k>1$, which accounts for the number of 
\phases that are used in each invocation of the procedure ($k$ controls the error bound). 
They return {\tt true} if a collision has been detected, {\tt false} otherwise.

Before the overall execution, each vertex generates a sequence $s$ of $k$ random bits ({\it u.a.r.}), each of which corresponds to a \phase.
Thus, if two nodes generate different sequences, all collisions between
them will be detected whatever the length of the computation.
Note that the same sequence will be used in every call of the procedure (i.e. it is fixed for every node); therefore, the random bits are not drawn from within the procedure itself.

\begin{algorithm}[h]
    \caption{\label{algo:emulateBcd}A Procedure to emulate a $B_{cd}$  in  the $BL$ model.}  
    \textbf{Procedure } {\tt EmulateB$_{cd}$inBL}($in: Integer\ k, Array$$<$$Boolean$$>\ s; out: Boolean\ collision$)
    
    $Boolean\ collision \gets false;$
    
    $Integer\ i \gets 0;$
    
    \Repeat{$i =  k$ }{
           
      \lIf{$s[i]$}{ beep in slot 1; listen in slot 2}
      
      \lElse{listen in slot 1; beep in slot 2}
      
      \lIf{another beep was heard}{$collision \gets true$}
       $i \gets i+1$
      
    }
    
    \textbf{End Procedure}
\end{algorithm}

\begin{algorithm}[h]
    \caption{\label{algo:emulateLcd}A Procedure to emulate a $L_{cd}$  in  the $BL$ model.}  
    \textbf{Procedure } {\tt EmulateL$_{cd}$inBL}($in: Integer\ k; out:Boolean\ beep, Boolean\ collision$)\\
    $Boolean\ beep \gets  false;$\\
    $Boolean\ collision \gets  false;$
    
    $Integer\ i \gets  0;$
    
    \Repeat{$i = k$ }{
       \Switch{slot}{ 
         \SetKwSwitch{Switch}{Case}{}{}{}{slots}{}{}{}%
        \uCase{1 and 2}{ 
          listen
        }
        \SetKwSwitch{Switch}{Case}{}{}{}{end of phase:}{}{}{}%
        \uCase{}{
          \If{a beep was heard in any slot}{$beep \gets  true$}
          \If{a beep was heard in both slots}{$collision \gets  true$}
        }
      }
      $i \gets  i+1$
    }
    
    \textbf{End Procedure}
 \end{algorithm}

The value of $k$ depends on the bound we require on the probability of error,
a straightforward adaptation of the above analysis
gives us the values of Lemma~\ref{lem:emule} relative to a {\em single time step}.

\begin{lemma}\label{lem:emule}
For any $\varepsilon>0$, any  $n>0$ and any  $c>2$:
\begin{enumerate} 
\item if $k=\lceil \log\left(\frac 1 \varepsilon\right) \rceil$, the procedures are correct for a given node with probability $1-\varepsilon$
\item if $k=\lceil \log\left(\frac n\varepsilon\right) \rceil$, the procedures are correct for all nodes with probability $1-\varepsilon$
\item if $k = \lceil c\log(n)\rceil$,  the procedures are correct for all nodes {\it w.h.p.}
\end{enumerate}
\end{lemma}

Observe that in general, the size of the network $n$ is not known to the nodes, which is an obstacle to achieving the second and third types of guarantee. 
However, it is reasonable in practice to assume that the nodes know an {\em upper bound} on $n$, e.g., when a network of wireless sensors is deployed. The upper bound may even be loose without much consequence: so long as it is polynomial in $n$, the slowdown factor remains logarithmic in $n$.

In a {\em complete} computation, each node may have a collision with any of its
neighbours. To ensure detection of all collisions, a single node must choose a
sequence amongst the $2^k$ possible ones different from the (at most) $\Delta$ sequences chosen by its
neighbours; that is,
different sequences must be chosen at the two ends of {\em all} edges.
The probability that this does not happen is at most $n\Delta/(2(2^k))$, which gives

\begin{lemma}\label{lem:emuleglobal}
For any $\varepsilon>0$, any  $n>0$, any  $c>2$ and any computation:
\begin{enumerate}
\item if $k=\lceil \log\left(\frac \Delta \varepsilon\right) \rceil$, the procedures are correct for a given node over the whole computation with probability $1-\varepsilon$
\item if $k=\lfloor \log\left(\frac {n \Delta}\varepsilon\right) \rfloor$, the procedures are correct for all nodes over the whole computation with probability $1-\varepsilon$
\item if $k = \lfloor c\log(n) + \log \Delta\rfloor$,  the procedures are correct for all nodes over the whole computation {\it w.h.p.}
\end{enumerate}
\end{lemma}

Finally, note that the emulation procedures should be used even when {\em listening},
in order for the nodes to remain synchronized with each other. 
Likewise, the procedures should not be interrupted even after a collision is detected in order to preserve synchrony.

\paragraph*{Resulting complexity.}
All algorithms presented in Section~\ref{sec:algorithms} can be adapted in $BL$ by replacing those beep or listen instructions which require collision detection by calls to {\tt EmulateB$_{cd}$inBL} or {\tt EmulateL$_{cd}$inBL}. In the worst case, all the slots need such an adaptation, resulting in a multiplicative logarithmic slowdown (according to the three different levels of guarantees stated in Lemma~\ref{lem:emuleglobal}).



\subsection{Optimality of the emulation}

In this section, we prove that the emulation procedures presented in Section~\ref{sec:emuler}
are essentially optimal (i.e. asymptotically and up to a constant factor), namely, we prove a $\Omega\left(\log n\right)$ lower bound on the number of slots
required to detect collision with high probability in some graphs called {\em wheels}.
A $(m,s)$-wheel,
illustrated in Figure~\ref{fig:wheel},
is a graph $W=(V,E)$  such that
$V=u_1,\ldots, u_{4ms}$, the edges $E$ are all the $(u_{i-1},u_i)$ (modulo $4ms$) 
plus $m$  spokes, that is edges $(u_{is},u_{(i+2m)s})$ ($1\leq i \leq 2m$),
where the wheel can be odd (all spokes with $i$ odd) or even (all spokes with $i$ even).
The even and odd $(m,s)$-wheels are isomorphic.
We consider only situations in which all vertices $u_{is}$  ($1\leq i \leq 4m$) are in the same state, a state in
which they wish to beep and all other vertices are in the same internal state, a state in
which they do not wish to beep. Thus vertices at the ends of spokes and no others must
conclude that there is a collision. The slot complexity of any algorithm which detects collision in such a graph with high probability is proven to be $\Omega(\log n)$. 

\begin{figure}
  \centering
  \begin{tikzpicture}[scale=1.7]
    \draw (0,0) circle (1);
    \draw (90:1) node[above] {$4ms$} -- (270:1.1) node[below] {$2ms$};
    \draw (70:.98) node[above right] {$1s$} -- (70:1.02);
    \draw (50:1) node[above right] {$2s$} -- (230:1) node[below left, inner sep=0pt] {$2ms+2s$};
    \draw (30:.98) -- (30:1.02);
    \draw (10:1) node[right,yshift=2pt,xshift=2pt] {$4s$} -- (190:1) node[left,xshift=-2pt] {$2ms+4s$};
    \draw (-15:1) node[right, xshift=2pt, rotate=-15] {$\vdots$};
    \draw (130:1) node[above left] {$2ms+2is$} -- (310:1) node[below right] {$2is$};
    \draw (-65:1) node[right, xshift=2pt, rotate=-70] {$\vdots$};
    \draw (105:1.1) node[above, rotate=115] {$\vdots$};
    \draw (250:.98) -- (250:1.02) node[below, xshift=-15pt] {$2ms+1s$};
    \draw (210:.98) -- (210:1.02);
    \draw (160:1) node[left, xshift=-2pt, rotate=-20] {$\vdots$};
  \end{tikzpicture}
  \caption{\label{fig:wheel} The wheel gadget used in the proof of optimality for emulation. }
\end{figure}
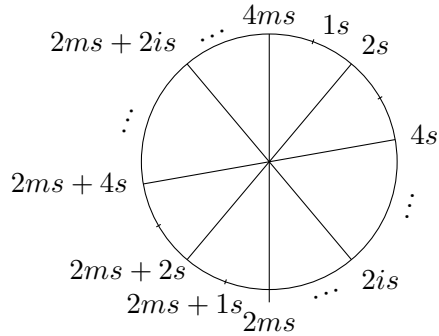

Considering a computation of a collision detecting algorithm on a wheel,
we define,
for any $t>0$,
$b_t^{i}$ as the signal (beep or not) from $u_i$ to all its neighbours at time $t$,
and, for any $t\geq 0$,
$B_t^{i}$ the sequence $b_1^{i}\cdots b_t^{i}$ .
Then, we define the event $E_t$ for a spoke $u_{is},u_{(i+2m)s}$ as follows: 
$$
E_t = 
(B_t^{{is}}=B_t^{{(i+2m)s}}).
$$
\begin{lemma}
\label{lem:wheel}
For any $t$ ($0 \leq t < s$), it holds that ${\mathbb P}r\left( E_t\right) \geq 2^{-t}.$
\end{lemma}

\begin{proof}
The proof proceeds by induction on $t$. Clearly the claim is true for $t=0$.
Suppose it is true for $t-1$.
The distribution of $B_t^{is}$ is determined by the initial states of
$u_{is-(t-1)} \cdots u_{is+(t-1)}$ since no $u$ at distance $t$ or more from $u_{is}$
can influence what happens at $u_{is}$ in time less than $t$.
Call this sequence of initial states $I$.

Hence for any sequence $X$ of $t-1$ beep/nobeeps, ${\mathbb P}r(b_t^{is} | B_{t-1}^{is} = X)$ is determined by $I$.
Hence for any $X$, ${\mathbb P}r(b_t^{is} | B_{t-1}^{is} = X)={\mathbb P}r(b_t^{(i+2m)s} | B_{t-1}^{(i+2m)s} = X)$
But by supposition $B_{t-1}^{is} =B_{t-1}^{(i+2m)s}$, so for any $X$,
${\mathbb P}r(b_t^{is} | B_{t-1}^{is} = X)={\mathbb P}r(b_t^{(i+2m)s} | B_{t-1}^{is} = X)$.
Let this conditional probability be $x$.
Then, ${\mathbb P}r(b_t^{is} =b_t^{(i+2m)s} | B_{t-1}^{is} = X) = x^2+(1-x)^2 \ge 1/2$
so that ${\mathbb P}r(E_t|E_{t-1} \wedge B_{t-1}^{is} = X) \ge 1/2$ and,
removing the conditioning, ${\mathbb P}r(E_t|E_{t-1}) \ge 1/2$ and so, by induction, ${\mathbb P}r(E_t) \ge 2^{-t}$.

%
%
%
\qed
\end{proof}

If $E_t$ holds for the spoke $(u_{is},u_{(i+2m)s})$,
we say that the spoke {\em fails} to break symmetry within time $t$.
This happens  with probability at least $2^{-t}$ and, if it happens,
the existence of the spoke has had no influence on the computation
up to time $t$.
In particular, whenever $u_{is}$ beeped, $u_{(i+2m)s}$ also beeped and so
neither has ever heard the other beep.
\begin{theorem}
\label{theorem::lb}
For any Monte Carlo algorithm $\cal A$ which 
detects collision in $W$, if $\cal A$ halts in less than $\log n/2$ slots
with probability greater than $3/4$ then for some situations in some wheels,
$\cal A$ gives incorrect results for some nodes with probability greater than $1/4$.
\end{theorem} 

\begin{proof}
The proof proceeds using the wheel gadget of Figure~\ref{fig:wheel} and Lemma~\ref{lem:wheel} to characterize the rate at which the symmetry induced by the spokes can be broken. For simplicity we
consider wheels $(m,s)$ where $s$ is a power of $2$ and $m=2^{2s-2}/s$ so that
$s=\log n/2$.
We consider a computation on this wheel without specifying whether it is
the odd or even wheel.
By Lemma~\ref{lem:wheel}, the probability that a given spoke $i$ breaks symmetry within time
$s-1$ is at most $1-2^{1-s} < exp(-2^{1-s})$
and this is independent for all spokes
so that the probability that every spoke breaks symmetry in the even case
in time $s-1$, is at most $exp(-2^{1-s}m)=exp(-2^{s+1}/s)<1/4$.
Hence 
the probability that the algorithm halts and
some spoke fails to break symmetry is greater than $1/2$.
If, in the even case,  spoke $i$ fails to break symmetry,
vertex $u_i$ hears the same signals from its neighbours in the odd and even cases
and, so, if it terminates the algorithm in this time, it has the same probability
of deciding collision in the two cases. Hence it gives a wrong result in one case
with probability at least $1/2$. Hence there is a vertex which gives a wrong result in
the odd or even case
with probability greater than $1/4$.

Finally, if an algorithm halts in time $o(\log n)$ with probability $\geq 3/4$,
for sufficiently large $n$
it halts in time less than $s$ and so its probability of
giving an incorrect result is at least $1/4$ for some initial conditions.
It follows that the same is true for any algorithm halting in expected time $o(\log n)$.\qed
\end{proof}

\begin{corollary}
The complexity of a Monte Carlo algorithm which detects collision with high probability in the $BL$ model is $\Omega(\log n)$.
\end{corollary}

\section{Complexity analysis}
\label{sec:analysis}

This section presents three complexity analyses, two of which are very generic. Taken together, the analyses cover (possibly indirectly) all the algorithms presented in this paper. We first present a new analysis of the MIS algorithm by Jeavons et al.\cite{JSX16}, proving that the number of phases before all nodes terminate is less than $76 \log n$ with high probability
As already discussed, the analysis in~\cite{JSX16} only guarantees termination within $e^{25} \log n$ phases, which is impractical, although the authors observe experimentally that the actual running time is much lower. Our analysis confirms this fact. The same analysis extends to the 2-hop version of MIS. Next, we analyse the colouring algorithm (Algorithm~\ref{algo:colouring}) and prove that it terminates within $O(\Delta + \log n)$ phases {\it w.h.p.}, where $\Delta$ is the maximum degree. This analysis extends in turn to the 2-hop variant of colouring (Algorithm~\ref{algo:2-hop-colouring}) by replacing $\Delta$ with $\Delta^2$ (neighborhood in the square of the graph), as well as degree computation (Section~\ref{sec:degree}). Finally, we analyse in Section~\ref{sec:analysis-colouring-K} the ($K+1$)-colouring algorithm with a known bound $K \ge \Delta$ (Algorithm~\ref{algo:colouring-K}). This algorithm is shown to terminate within $O(K\log n)$ phases {\em w.h.p.}, and the analysis extends to $O(K^2\log n)$ phases in the 2-hop variant (again, due to considering the square of the graph).

\subsection{Common definitions and notations}
The three analyses presented here share a number of common traits. First, all processes consist in a competition among neighbours for producing exclusive beeps (with various consequences). 

\paragraph{Genericity of the analyses.} The principle at stake in the first analysis (MIS algorithm) is quite generic: nodes compete to produce exclusive beeps, and when one succeeds, this node and all its neighborhood terminate. In the second analysis (colouring algorithm), the principle is similar except that only the node producing the exclusive beep terminates, the others continuing competition. The third analysis is more specific to the problem considered ($(K+1)$-colouring). In all cases, we analyse the time it takes until {\em all} nodes have terminated.

\paragraph{Residual graph.} We call {\em residual graph} at a given phase, the graph induced by the nodes that are still contending for producing exclusive beeps. As such, it is initially the same as the original graph, and becomes eventually empty. At any point of the analysis, we denote by $d$ the current degree of a node in the residual graph. We say that a node {\em survives} so long as it is in the residual graph.

\paragraph{Adaptive probability.} All algorithms consider an adaptive probability (see pattern in Section~\ref{sec:patterns}). The use of an adaptive probability is what makes the algorithms efficient. However, it is what makes their analysis more complicated. The first two algorithms are analysed using a doubling/halving pattern, that is, when $p$ is increased or decreased, it is so by a factor $2$, {\em with upper bound} and {\em initial value} $1/2$. Other general strategies exist, such as incrementing or decrementing the denominator of $p$ (see e.g.~\cite{CMRZ16c}), they are not considered here. In the third analysis (which is less generic), we use part of the input as an ingredient for the update of the adaptive probability.

\paragraph{Local beeping probability.} At a given phase, we write $p_v$ for the probability that a given node $v$ beeps and we use $q_v = \sum_{u\in N(v)}p_u$ as a trivial upper bound on the probability that at least one of its neighbor beeps. We omit the subscript $v$ when it is clear from the context. Intuitively, the higher the value $p$ and the lower $q$ at a node, the more likely is it that this node produces an exclusive beep. 

\subsection{MIS algorithm (Jeavons et al.~\cite{JSX16})}
\label{sec:analysis-MIS}

In this algorithm, when a node produces an exclusive beep, it enters the MIS and terminates. All of its neighbors which were still active decide not to be in the MIS and terminate as well. We prove that this algorithm terminates in less than $76 \log n$ phases with probability $1-o(n^{-1})$. Note that the asymptotic order is already optimal, since the $\Omega(\log n)$ lower bound for colouring in the message passing model with constant size messages~\cite{KOSS} applies to the MIS (see~\cite{MRSZ10} for details), and this model is stricly stronger than any version of the beeping model.

\paragraph{Outline of the analysis.} The main result (Theorem~\ref{th:MIS}) is a bound on the probability that a given node, with given $p$ and $q$, remains active over the next $t$ phases (that is, neither this node or one of its neighbors beep exclusively). This proof relies on an intermediate result (Lemma~\ref{lem:MIS}) which establishes that whatever the events occurring in a given phase in the surroundings of a given node $v$, it is sufficient to focus on the special case that $q_v$ is halved in the next phase (thus conditions improve as time passes). Finally, based on Theorem~\ref{th:MIS}, we conclude in Corollary~\ref{cor:MIS} that a certain number of phases ($76 \log n$) are sufficient to guarantee that the execution terminates at a given node with probability $1 - o(n^{-2})$, and thus everywhere with probability $1 - o(n^{-1})$.

\paragraph{Detailed proof.}

At any time, the probability that a node survives in the next $t$ phases depends on the current value of its variables $p$ and $q$. We define $t_0 = 3\,l(q) - 2 \log p$, where $l(q)$ equals $\log (5\max\{q,1/5\}) = \max\{\log(5q),0\}$. (This technical adjustment prevents the occurrence of negative values later on.) Our main result is the following theorem.
\begin{theorem}
\label{th:MIS}
For any $t\ge 0$ and node $v$, its probability of remaining
active after the next $t$ phases is at most $\alpha^{t_0-t}$ with $\alpha=2^{1/36} \approx 1.01944$.
\end{theorem}
\begin{proof}
The proof will be by induction on $t$. We have $t_0\ge 2$, so that if $t=0$,
$\alpha^{t_0-t} >1$ and the claim is trivially true.
Now, let $t>0$. After one phase which does not add $v$ or a neighbour to
the MIS we have by induction that the probability of remaining active
for the following $t-1$ phases is at most $\alpha^{t^\prime_0-t+1}$
where $t^\prime_0$ is the new value of $t_0$, namely $ 3
l(q^{\prime})- 2\log p^\prime$.  So we conclude that the probability
of survival over the next $t$ phases is upper bounded by the weighted mean of $\alpha^{t^\prime_0-t+1}$ if $v$ survives the first phase and $0$
otherwise. We refer to this mean as the {\it bound} and note that it
is dependent on what happens outside  $\overline{N}(v)$
as well as the choices of all nodes in $\overline{N}(v)$.
\end{proof}

The rest of the proof relies on an intermediate Lemma, which we prove next. Let us first define a useful concept in this direction.

\begin{definition}[Inhibition] A node is said to be {\em inhibited} in a phase if at least one of its neighbours beeps in that phase.
\end{definition}

We will decline a number of cases and subcases, with frequent operations on exponents and logarithms. In particular, note that $\alpha^{3\log q}=q^{3\log\alpha}=q^{1/12}$.
\begin{lemma}
  \label{lem:MIS}
The bound is maximised when what happens outside the neighbourhood of
$v$ is that every neighbour $u$ of $v$ is inhibited from joining the
MIS by some external neighbour beeping and no neighbour of $v$
becomes inactive through another node (outside $\overline{N}(v)$)
joining the MIS.
\end{lemma}
\begin{proof} 
Clearly a node outside $\overline{N}(v)$ joining the MIS can only
affect the bound by reducing $q$ which reduces the bound.
Consider any external behaviour $E$ in which some $u$ is not
inhibited; we will show that the bound is increased or unchanged if
the behaviour is changed to $E^\prime$ in which $u$ is inhibited and
there is no change for any other neighbours of $v$.  (In a given graph
there may be no such $E^\prime$ but we consider the maximum possible
over any graph containing the neighbourhood $\overline{N}(v)$.)  We
consider fixed beeping decisions of all nodes in $\overline{N}(v)$
except $u$ and show that with these decisions $E^\prime$ gives a value
of the bound greater than or equal to that of $E$.  We consider three
cases:

\begin{itemize}
\item Some neighbour of $v$ which is neither $u$ nor a neighbour of
  $u$ enters the independent set: Note that this is determined 
  by the fixed beeping decisions and the external behaviour other than
  as it affects $u$. Hence this happens for $E$ iff it also happens
  for $E^\prime$ and in each case the bound is $0$.
\item Some neighbour of $u$ in $\overline{N}(v)$ beeps: $p_u$ will
  be halved whether or not $u$ is inhibited by $E^{\prime}$ and so
  both $p^\prime$ and $q^\prime$ and the probability of survival are
  the same for $E$ and $E^\prime$. The bound is identical in the two
  cases.
\item Otherwise: 
Let the value of $p^\prime$ be $p_0$ if $u$ does not beep and $p_1$ if
$u$ does beep ($p_1 \le p_0$).
Let the value of $q^\prime$ be $q_0$ if $u$ does not beep and is not
inhibited, $q_1$ if it beeps and is inhibited and $q_2$ if it does not
beep and is inhibited.  Note that if $u$ beeps and is not inhibited,
$u$ enters the independent set and $v$ does not survive.  We have
$q_1\ge q_0/4$ since, at most, $u$'s beeping can result in a node
$w$ halving $q_w$ when otherwise it would have doubled it. Similarly
$q_2\ge q_0/4$ and $q_2 \geq q_0-3p_u/2$ since the inhibition results in
$p_u$ being halved rather than potentially doubled.

The bounds are thus $p_u\alpha ^ { 3 l(q_1)- 2 \log(p_1)-t+1}+
(1-p_u)\alpha ^ { 3 l(q_2)- 2 \log(p_0)-t+1}$ in the inhibited case
and $(1-p_u)\alpha ^ { 3 l(q_0)- 2 \log(p_0)-t+1}$ in the
uninhibited case.  We claim that the ratio of the inhibited bound to
the uninhibited is at least $1$.  Since $p_1\le p_0$, this ratio is at least equal to $\frac{p_u\alpha^{ 3 l(q_1)}+(1-p_u)\alpha^{ 3 l(q_2)}} {(1-p_u)\alpha^{ 3 l(q_0)}}$.\\ 
Remember that $p_u$ is a power of $1/2$.  We
consider four subcases:
\begin{itemize}
\item $q_0 \le 1/5$: $l(q_1)=l(q_2)=l(q_0)=0$ and the ratio
$\ge (p_u+1-p_u)/(1-p_u) > 1$.

\item $1/5 < q_0$ and $p_u\ge 1/8$: We use the bounds $q_1 \ge q_0/4$ and
$q_2 \ge q_0/4$ giving that the
ratio is at least
$(p_u+1-p_u)\alpha^{-6}/(1-p_u)$
$= \alpha^{-6}/(1-p_u) \ge \alpha^{-6}(8/7) \ge 1$.
\item $1/5 <q_0 \le 4/5$ and $p_u \le 1/16$: We use the bounds $q_1
  \ge q_0/4$ and $q_2 \ge q_0 -3p_u/2$ and the fact that for $0<x \le
  15/32$, $(1-x)^{1/12} > 1-4/3(x/12)$ so that the ratio is at least
  $p_u\alpha^{-6}/(1-p_u) + (1-3p_u/2q_0)^{3\log\alpha}$ $\ge
  p_u\alpha^{-6} + (1-15p_u/2)^{1/12}$ $\ge p_u\alpha^{-6} +
  (1-(15p_u/2)/12\times (4/3))$ $\ge 1 + p_u(\alpha^{-6} - 5/6) > 1$.

\item $q_0 > 4/5$ and $p_u\le 1/16$: Using the same bounds as in the
  previous subcase the ratio is greater than
  $\frac{p_u}{1-p_u}\alpha^{ -6} + \alpha^{ 3 (l(q_0-3p_u/2)-l(q_0))}$
  $>\frac{p_u}{1-p_u}\alpha^{-6 } + \alpha^{ 3
    (l(4/5-3p_u/2)-l(4/5))}$ and this is the bound already used for
  the case with $q_0=4/5$ and the same value of $p_u$ and so is
  greater than or equal to $1$.
\end{itemize}
\end{itemize}
This ends the proof that $E^\prime$ gives a value for the bound at
least as great as that for $E$. The lemma is then proved by a simple
induction on the number of uninhibited nodes. \qed
\end{proof}

\begin{remark}
  In the situation described by Lemma~\ref{lem:MIS}, the value of $q$ is halved
unless $v$ joins the MIS.
\end{remark}

We now return to the inductive proof of Theorem~\ref{th:MIS}.
From Lemma~\ref{lem:MIS} we 
know that the survival probability is at most the weighted sum of
$\alpha^{ 3  l(q/2) - 2  \log p^{\prime}-t+1}$ if $v$ survives the first phase
and $0$ otherwise.
We consider the following five cases, which cover all possibilities.

\begin{enumerate}
\item $q \ge 2/5$: We have $l(q/2)= l(q)-1$ and $p^\prime \ge p/2$ giving
$P(survival) \le \alpha^{ 3  (l(q)-1)- 2 (\log p-1)-t+1}$
$=\alpha^{ 3  l(q)- 2 (\log p)-t}$
as claimed.

\item $1/5 \le q < 2/5$ and $p<1/2$: The probability that a neighbour
  of $v$ beeps is less than $q$ so that $p_v$ is doubled with
  probability at least $1-q$ and halved in the remaining cases. In all
  cases $l(q/2)=0$.  Hence $P(survival) \le
  \alpha^{-2\log(p)-t+1}((1-q)\alpha^{- 2 } +q\alpha^{ 2 })$ and our
  claim is that it is at most
  $\alpha^{3\log(5q)-2\log(p)-t}$. That is the claim is valid
  since $(1-q)\alpha^{- 1 } +q\alpha^{ 3 }\le\alpha^{3\log(5q)}$ in
  the range $1/5 \le q < 2/5$.  (It is valid at $q=1/5$ since
  $4\alpha^{-1}+\alpha^3<5$ and at $q=2/5$ since
  $3\alpha^{-1}+2\alpha^3<5\alpha^3$; between these two limits, the
  left hand side is linear and the right hand side
  ($(5q)^{3\log\alpha}$) has a negative second derivative so the
  inequality holds there also.)

\item $1/5 \le q < 2/5$ and $p=1/2$: With probability greater than
  $1-q$ no neighbour of $v$ beeps and then $v$ has probability $1/2$
  of entering the independent set; otherwise $p_v$ remains $1/2$. On
  the other hand, if a neighbour does beep, $p_v$ becomes $1/4$.  In
  all cases $l(q/2)=0$.  Thus the probability of survival $\le
  \alpha^{2-t+1}((1-q)/2+q\alpha^2)$ and the claim is that it is at
  most $\alpha^{3\log(5q)+2-t}$. That is the claim is valid if
  $(1-q)\alpha/2+ q\alpha^{ 3 }\le\alpha^{3\log(5q)}$ a weaker
  condition than in the previous case.

\item $q < 1/5$ and $p<1/2$: The probability that a neighbour of $v$
  beeps is less than $1/5$ so that $p_v$ is doubled with probability
  at least $4/5$ and halved in the remaining cases. In all cases
  $l(q)$ decreases or is unchanged.  Hence $P(survival) \le \alpha^{ 3
    l(q) - 2 \log(p)-t+1} ((4/5)\alpha^{- 2 } +(1/5)\alpha^{ 2 })$
  and this is less than $\alpha^{ 3 l(q) - 2 \log p-t}$ as claimed,
  again since $4\alpha^{- 1 }+ \alpha^{ 3 }<5 $.

\item $q < 1/5$ and $p=1/2$: With probability greater than $4/5$ no
  neighbour of $v$ beeps and then $v$ has probability $1/2$ of
  entering the independent set; otherwise $p_v$ remains $1/2$. On the
  other hand, if a neighbour does beep, $q$ decreases and $p_v$
  becomes $1/4$.  Hence $P(survival) \le (2\alpha^{ 3 l(q/2)- 2
    \log(1/2)-t+1}+ \alpha^{ 3 l(q/2) - 2 \log(1/4)-t+1})/5$ $\le
  \alpha^{ 3 l(q)- 2 \log(1/2)-t+1}(2+\alpha^{ 2 })/5$ which is at
  most $\alpha^{ 3 l(q) - 2 \log(1/2)-t}$ as claimed since
  $2+\alpha^{ 2 }<5\alpha^{-1} $.
\end{enumerate}
\qed

We can now conclude on the termination time of the algorithm.
\begin{corollary}\label{cor:MIS}
The number of phases taken by the MIS algorithm is less than $76\log n$ {\it w.h.p.}
\end{corollary}
\begin{proof}
Since initially $p_v=1/2$ and $q_v <n/2$ where the graph has $n$
nodes, we conclude that $t_0 < 3\log(5n/2)-2\log(1/2)<3\log
n+6$ so that after $t=76 \log n$ phases, every node $v$ has
probability at most $\alpha^{-73 \log n + 6} = o(n^{-2})$ of still being active. A final union bound extends this to all the nodes with probability $1-o(n^{-1})$.\qed
\end{proof}

\subsection{Colouring without knowledge}
\label{sec:analysis-colouring}

In this algorithm, when a node produces an exclusive beep, it takes the current round number as a colour and terminates (see Algorithm~\ref{algo:colouring} for details). Unlike the MIS algorithm, the neighborhood of this node keeps contending afterwards, which leads to a longer execution time and (at least in our analysis) a dependency on the maximum degree $\Delta$. Hence, we prove that this algorithm terminates within $O(\Delta + \log n)$ phases with probability $1-o(n^{-1})$. Finally, we prove a matching $\Omega(\Delta + \log n)$ lower bound that establishes that our algorithm (and analysis) is optimal.

\begin{theorem}
\label{theorem::analyse_coloration}
For any graph $G = (V, E)$ with $|V|=n$  
and maximum degree $\Delta$, Algorithm~\ref{algo:colouring} colours all the nodes within $O(\Delta +\log n)$ phases with probability $1 - o\left(n^{-1}\right)$.
\end{theorem}


\paragraph{Outline of the proof.} 
The process is first understood from an inductive point of view ({\em i.e.} from one phase to the next), then the overall time complexity is derived.
More precisely, we first define and analyse a measure $M$ of the distance from an arbitrary configuration of these values to the favorable case where $p=1/2$ and $q\le1/2$ (or $d=0$). 
Intuitively, given a node $v$, we expect both $p$ and $q$ to decrease initially until $q<1/2$,
after which $p$ will re-ascend until it is at least close to $1/2$
and then $d$ will start to descend. 
The initial value of $M$ is chosen in such a way that it decreases at least by $1$ on average in every phase. After this favorable configuration has been reached, we show that it takes a certain (logarithmic) number of further phases to terminate at $v$ with probability
$o(n^{-2})$ and thus everywhere with probability $o(n^{-1})$.

\subsubsection{6.3.1~~~From one phase to the next}~\smallskip

We define a measure $M$ of the distance from a given situation to the goal
where $p=1/2$ and $q_v\le1/2$ (or $d=0$). We will prove that $M=-\log(p) + f(q) +10d$ is a sufficient value, 
where $f$ is the function defined as follows: If $q \le 1$, then $f(q)=4q$. Otherwise, it is the piecewise linear approximation to $2 \log 4q$, i.e. $f$ is interpolated linearly between 
$f(2^i)=2i+4$ and $f(2^{i+1})=2i+6$. This particular definition guarantees the following convenient properties:
\begin{itemize}
\item $f(q)$ is continuous for $q>0$,
\item except at powers of $2$, $f$ is differentiable with derivative $\le~4$,
\item $f(q)-f(q/2)~=~2$ for $q\ge 1$,
\item $f(q)-f(q/2)~=~2q$ for $q\le 1$.
\end{itemize}
The goal is to show that in any phase, the mean decrease in $M$ is at least $1$.
Thus, after a number of phases equal to the initial $M$ ($\le$ $1+2\log(2d)+10d$),
$M$ is reduced on average
to $0$ unless the algorithm has already terminated at $v$. 

In fact, $M$ can decrease indefinitely in one phase (because two or more of $v$'s neighbours beep exclusively in the phase) and we want to apply later a theorem on martingales with bounded variation. Accordingly, we define another random variable $M^*$ which dominates $M$ and has the desired properties. The r.v. $M^*$ is initially equal to $M$ but its changes may be slightly different in the following ways:
\begin{itemize}
\item if $M$ decreases by more than $11$, then $M^*$ decreases by exactly $11$;
\item if $v$ beeps exclusively, $M^*$ is decreased by just $10$
whatever the values of $p$ at $v$ and its neighbours;
\item in a phase where $v$ has already terminated, $M^*$ is decreased by $1$.
\end{itemize}
This ensures that:
\begin{itemize}
\item  if the algorithm has not terminated at $v$, $M^*\ge M$, meaning that $M^*$ dominates $M$;
\item if $M^* \le 0$, the algorithm has terminated at $v$;
\item $M^*$ decreases at each phase by a value in $[-3\ldots 11]$
(since $\log p$ cannot decrease by more than $1$ and $f(q)$ cannot increase by more than $2$).
\end{itemize}

\paragraph{Managing dependencies among neighbors.}

Given a node $v$, we denote by $u_i$ ($1\le i \le d$) its neighbours in the current residual graph. 
In a given phase any $u_i$ has a well defined probability $a_i$ that none of its neighbours beep.
These probabilities are far from being independent. We will argue that,
except for cases where $p_{dec}$ (the probability of a decrease in $d$) is at least $2/5$,
the average decrease in $M^*$  is always minimised when all $a_i=0$.

Consider a situation where $p_{dec}<2/5$ and some $a_i>0$. We can decrease $a_i$
to $0$ with no change to the other $a_j$ by adding an infinite number of nodes
adjacent to $u_i$ but to no other node in $\overline{N}(v)$.
This will change the average increase/decrease in $q$ and $d$;
$q_i$ will be halved instead of doubled with probability $a_i$, decreasing $q$
by $3q_i/2$ and so decreasing $f(q)$ by at most $6a_iq_i$ on average.
The probability of a decrease in $d$ is decreased by $a_i$ times the probability
that $u_i$ beeps and no other $u_j$ beeps exclusively. This last probability
is the product of $q_i$ and the conditional probability that no other $u_j$ beeps exclusively
given that $u_i$ beeps. But,
since the probabilities of $u_i$ beeping and of some $u_j$ ($j\not=i$) beeping exclusively
are negatively correlated or independent,
this conditional probability is at most the
unconditional probability that no $u_j$ ($j\not= i$) beeps exclusively and so greater than
the probability that no $u_j$ beeps exclusively, namely $1-p_{dec} \ge 3/5$,
giving an average decrease in $d$ (respectively $M^*$)
reduced by more than $3a_iq_i/5$ (resp. $6a_iq_i$).

Thus the decrease in the measure is decreased more
by the $d$ component than it is increased by the change in $q$.

Repeating this process at most $d$ times we arrive at a situation with a smaller
mean decrease in $M^*$ than the initial one and either $p_{dec}\ge 2/5$ or all $a_i=0$
so, to lower-bound the decrease in $M^*$, we need only consider such situations.

\paragraph{The mean decrease in $M^*$.}

We consider cases depending on the value $q$.
First note that if $p_{dec}\ge 2/5$,
the mean decrease in $M^*$ is at least $10(2/5)-2-1=1$.
So in the other cases we suppose that all $a_i=0$ so that $q$ is halved.

In the case where $q<1$, we need to consider what happens when no $u_i$ beep.
If $p<1/2$ this is that $p$ increases, decreasing $M^*$ by $1$; if $p=1/2$ it is that,
with probability $1/2$, $v$ beeps exclusively so that
$d$ decreases by $1$, decreasing $M^*$ by $10$ so that on average $M^*$ decreases by $5$.
Accordingly, we suppose that the former happens.
\begin{itemize}
\item {$q \ge 1$}: 
$q$ decreases to $q/2$, reducing $f(q)$ by $2$ and $\log p$ can decrease by at most $1$
so that $M^*$ is decreased by at least $2-1 = 1$.
\item {$q < 1$}: 
$q$ decreases to $q/2$, decreasing $f(q)$ by $2q$ while $p$ doubles with probability
at least $1-q$ (and halves with probability at most $q$).
This gives a mean decrease in $M^*$ of at least $2q +(1-2q) = 1$.
\end{itemize}

\subsubsection{6.3.2~~~Overall time complexity}~\smallskip

We define the sequence of r.v.'s $(M_k)_{0\leq k\leq t}$ as follows $M^*_0= M_0$ and 
for any $k\geq 1$, $M_k$ is the value of $M^*$ after time $k$. We also define the sequence $(G_k)$ of residual graphs, where $k$ is the phase number.

Then for any $k\geq 1$:
\begin{eqnarray}
{\mathbb E}\left( M_k\mid G_1,G_2, \cdots,G_{k-1}\right) \leq M_{k-1} - 1.
\end{eqnarray}
Hence, $(M_k)_{k\geq 0}$ is a super-martingale with respect to $(G_k)_{k\geq 0}$. 

We define the  r.v. $D_k = M_k - M_{k-1}$ for any $k\geq 1$ and we denote $\mu = {\mathbb E}\left(D_k\right)$. We also introduce the r.v. $D'$ so as to retain the range of possible values of $D$ but
have mean exactly $-1$:
$$
D'_k = -\frac{4}{\mu-3}D_k + \frac{3\mu+3}{\mu-3}. 
$$
Then, it is easy to see that ${\mathbb E}\left(D'_k\right) = -1$ and ${\mathbb P}r\left(-11\leq D'_k\leq 3\right) =1$.

Now, define the r.v. $(M'_k)_{k\geq 0}$ as follows: $M'_0 =M_0$ and for any $k\geq 1$, $M'_k = M'_{k-1}+D'_{k} +1$. Then $(M'_k)_{k\geq 0}$ is a martingale with respect to $(G_k)_{k\geq 0}$.   

We apply Theorem 18 of \cite{AK67} to our martingale
$M'_t$ with expectation $M_0$.
Since the increments $(D'_k+1)_{k\geq 0}$ are in $[-10..4]$ and have mean $0$, their
variance is upper bounded by the case of a distribution with
values $-10$ and $4$ with probabilities $2/7$ and $5/7$ respectively,
giving variance of $40$ and maximum discrepancy from the mean of $10$.
Applying the theorem with $t=2M_0+174\ln n$ and $\lambda=t-M_0=M_0+174\ln n$,
we see that ${\mathbb P}r(M_t \geq 0)$ is less than ${\mathbb P}r\left( M'_t\geq t\right)$ which is at most:
$$
e^{\left(-\lambda^2/2(40t+10\lambda/3)\right)}.
$$ 
We have $\lambda=t-M_0$ so that $\lambda^2 > t(t-2M_0) = 174 \ln t$ and also $\lambda=M_0+174 \ln n > 174\ln n$
so that $\lambda^2 > 174 \lambda \ln n$.
Adding these two with weights of $6/13$ and $1/26$ gives
$$\lambda^2/2 > 174 \ln n(6t/13 + \lambda/26)$$
which gives
$$\lambda^2/2(40t+10\lambda/3) > 261/130 = 2+1/130$$
so that ${\mathbb P}r(M_t \geq 0) = o(n^{-2})$.

Then taking $t=2M_0+174\ln n$ is sufficient. Since $M_0=M \le 1+2\log(2d) + 10d$, taking $t=20\Delta + 180\ln n$ proves Theorem~\ref{theorem::analyse_coloration}
(allowing for the mixture of bases of the logarithms).\qed

\subsubsection{6.3.2~~~$\Omega(\Delta + \log n)$ Lower bound}~\medskip\\
We now establish that $\Omega(\Delta + \log n)$ slots are actually required for the colouring problem in the class of Las Vegas beeping algorithms. On the one hand, it is already known that $\Omega(\log n)$ rounds are needed to colour the nodes of a ring in the synchronous message passing model with constant size messages~\cite{KOSS}. Since this model can trivially simulate any of the beeping models, the bound applies. 
We now establish that $\Omega(\Delta)$ are needed as well in an infinite family of graphs. More precisely, we show that $\Omega(n)$ 
slots are required for colouring the  complete graph $K_n$
with a Las Vegas Algorithm even in $B_{cd}L_{cd}$.

\begin{lemma}\label{lowerbound}
Colouring $K_n$ with a Las Vegas beeping algorithm takes $\Omega(n)$ slots.
\end{lemma}
\begin{proof}[By contradiction]
Let $\cal A$ be such an algorithm and let $\cal E_{A}$ be an execution of $\cal A$ that terminates in less than $n$ slots in the complete graph $K_n$. Then it holds that at least one node, say $v$, never beeped exclusively. Let $\cal E_{A}'$ be another execution of $\cal A$, this time in the complete graph $K_{n+1}$ composed of the same nodes plus $v'$. Let all the nodes behave as they did over $\cal E_{A}$ and let $v'$ act exactly like $v$. Since $v$ never beeped alone in $\cal E_{A}$, the same is true in $\cal E_A'$ and for $v'$, making the two executions indistinguishable (two beeps are indistinguishable from three). Hence, the nodes in $\cal E_A'$ terminate as in $\cal E_A$, 
and $v$ has the same colour as $v'$ which is a contradiction.
\end{proof}


\begin{theorem}
  Las Vegas colouring takes $\Theta(\Delta + \log n)$ slots in the strongest beeping model ($B_{cd}L_{cd}$).
\end{theorem}

\subsection{$(K+1)$-Colouring knowing $K \ge \Delta$.}
\label{sec:analysis-colouring-K}

We analyse here the variant defined at the end of Section~\ref{sec:colouring-K}, which corresponds partially to Algorithm~\ref{algo:colouring-K}, with a different kind of adaptive probability. Recall that the algorithm consists for a node to take the current round number (modulo $K+1$) as colour when it produces an exclusive beep (then it terminates). The adaptive probability is managed as follows. We call a {\em cycle} a sequence of $K+1$ phases. In the beginning of each {\em cycle}, every node updates its beeping probability $p$, setting it to $1/(2|Colours|)$, where $|Colours|$ is the number of remaining colours (that is, $K$ minus the number of colours already taken by a neighbor). In this context, we prove that the number of phases is $O(K\log n)$ with probability $1-o(n^{-1})$.

\begin{theorem} 
Let $G$ be a graph of size $n$ and $K \ge \Delta$ an upper bound on the maximum degree, then this algorithm computes a $(K+1)$-colouring of $G$ within 
$O(K \log n)$ phases with probability $1-o(n^{-1})$.
\end{theorem} 
\begin{proof}

Let $P_k$ be the probability that
node $v$ survives uncoloured over $k$ cycles (where $k$ is unrelated to the bound $K$). We will use the following symbols recurrently, with given domains of definition:
\begin{itemize}
\item $i$ ranges over $1..k$,
\item $c$ ranges over the $C_i$ colours possible for $v$ at the start of cycle $i$,
\item $u$ ranges over the neighbours of $v$ still uncoloured at the start of cycle $i$,
\item $p_u(i,c)$ is the probability that $u$ beeps for colour $c$ in cycle $i$.
\end{itemize}
First we consider the probability $p$ that $v$ survives uncoloured in a single phase
using a colour $c \in Colours(v)$. Then:
\begin{eqnarray}
p & = & {\mathbb P}r\left(v{\rm~does~not~beep~at~colour~}c{\rm~in~cycle~}i\right)\nonumber\\
   & + & {\mathbb P}r\left(v{\rm~does~beep~and~some~neighbour~}u{\rm~also~beeps}\right),\nonumber
\end{eqnarray}   
but ${\mathbb P}r\left(v{\rm~does~beep}\right) = 1/2C_i$
and the beeping probabilities of $v$ and its neighbours
are independent giving:
\begin{eqnarray}
p & =  & \left(1-1/2C_i\right) + {\mathbb P}r\left({\rm some~neighbour~beeps}\right)/2C_i\nonumber\\
   & =    & \left(1-1/2C_i\right)\left(1+{\mathbb P}r\left({\rm some~neighbour~beeps}\right)/(2C_i-1)\right)\nonumber\\
   & \le  & \left(1-1/2C_i\right)\left(1+\sum_u p_u(i,c)/(2C_i-1)\right).\nonumber
\end{eqnarray}
After the first phase, $p_u(i,c)$ and $C_i$ are random variables
dependent on what has happened so far,
and we consider the tree of all possible executions up to $k$ cycles,
where each tree node has its own value of $p$.
It is easily shown by induction that
$P_k$ is upper bounded by the maximum over all paths in this tree
of the product of the values of $p$ along the path.
We fix a path which gives this maximum and
bound the product for this path.
We have the probability of surviving cycle $i$
 $\le  (exp(-1/2) \times \prod_c  (1+\sum_u p_u(i,c)/(2C_i-1)))$
 $\le  exp(-1/2 +  \sum_c \sum_u p_u(i,c)/(2C_i-1))$
and so
$P_k \le  exp(-k/2 + \sum_i \sum_c \sum_u p_u(i,c)/(2C_i-1))$.




In cycle $i$, $v$ has $C_i$ colours available and so has less than $C_i$ neighbours;
each neighbour $u$ has $\sum_c p_u(i,c) \le 1/2$, giving, for this cycle,
$\sum_u \sum_c p_u(i,c)/(2C_i-1)\le 1/4$ so that $\sum_i \sum_u \sum_c p_u(i,c)/(2C_i-1) \le k/4$.


Hence $P_k \le  exp(-k/4)$ and after 
$9\ln n$ cycles, $v$ has probability $o(1/n^2)$ of
remaining uncoloured and the graph has probability $o(1/n)$ of having any uncoloured
node.\qed
\end{proof}

\subsection{Two-hop variants of the algorithms}
As explained in Section~\ref{sec:algorithms}, the 2-hop variants of both colouring algorithms (with or without knowledge) come to the same essential operations as the 1-hop variant, but performed in the stronger model $B_{cd}L_{cd}$, in the square of the graph and using (a constant number of) additional slots in each phase for reporting peripheral collisions. This being said, the complexity of these algorithms remains essentially the same, replacing only the $\Delta$ term with $\Delta^2$ (or replacing $K$ by $K^2$), due to acting in the square of the graph. By an analogy already discussed in Section~\ref{sec:algorithms}, the resulting $O(\Delta^2 + \log n)$ complexity also applies to the computation of the degree. Finally, the same arguments also apply to the MIS algorithm, but in this case, since there is no dependency on $\Delta$, the complexity of the 2-hop variant remains $O(\log n)$. In fact, the complexity remains within the same complexity of $76 \log n$ without additional penalty, because the maximum number of 2-hop neighbors cannot exceed $n$ (which is the same bound as for the number of one-hop neighbors).

\section{Further Related Work}
\label{sec:related}

This section provides further related work on beeping algorithms and algorithms for radio networks.
As explained by Chlebus~\cite{C01}, in a radio network, a node can hear
a message only if it was sent by a neighbour and this neighbour was the only
neighbour that performed a send operation in that step. If no message
has been sent to a node then it hears the background noise. If a node
$v$ receives more than one message then we say that a collision occurred at the
node $v$ and the node hears the interference noise. If the nodes
of a network can distinguish the background noise from the interference noise,
then the network is said to be with collision detection, otherwise it is
without collision detection (see for example  the Wake-up problem, MIS problem or election in radio networks 
in \cite{GPP01,Moswa,CGK07,JK15,LMR07} where nodes cannot distinguish 
between no neighbour sends a message and at least
two neighbours send a message; see also the broadcasting problem
in radio network in \cite{GHK13} where nodes can distinguish between
no neighbour sends a message, exactly one neighbour sends a message and 
at least two neighbours send a message).
In this context, an efficient randomised
emulation of a single-hop radio network with collision detection on
multi-hop radio network without collision detection is presented and
analysed in \cite{BGI91}.
To summarise, detecting  a collision in a radio network is to be able to distinguish
between $0$ message and at least $2$ messages while detecting a collision
in the beeping model is to be able to distinguish between $1$ message
and at least $2$ messages.

Despite this difference, some of our algorithms use
similar ideas to those used for initialising a packet radio network
\cite{HNO99} or for election in a complete graph with wireless
communications \cite{BW12} (Algorithm 50, p. 132).  The impact of
collision detection is studied in \cite{Schneider10,KP13}, where it is
proved that performances are improved, and in certain cases the
improvement can be exponential.  The complexity of the conflict
resolution problem (where the goal is to let every active node use the
channel alone (without collision) at least once) is studied in
\cite{HuangM13} (they assume that nodes are identified), and an
efficient deterministic solution is presented and analysed.

Regarding the MIS and colouring problems, general considerations and many examples of Las Vegas distributed
algorithms  can be found in \cite{Peleg}.
The computation of a MIS has been the object of extensive research on
parallel and distributed complexity in the point to point message
passing model \cite{Alonco,Luby} \cite{Awerbuchco,Linial}; Karp and
Wigderson \cite{KarpW} proved that the MIS problem is in NC.
Some links with distributed graph colouring and some recent results on
this problem can be found in \cite{Kuhnco}.  The complexity of some
special classes of graphs such as growth-bounded graphs is studied in
\cite{Kuhnmo}.  Results have been obtained also for radio networks
\cite{Moswa}.  A major contribution is due to Luby \cite{Luby}.  He
gives a Las Vegas distributed algorithm.  The main idea is to obtain
for each node a {\it local total order} or a {local election} which
breaks the local symmetry and then each node can decide locally
whether it joins the MIS or not.  Its time complexity is $O(\log n)$
and its bit complexity is $O(\log^2 n).$ Recently, a Las Vegas
distributed algorithm has been presented in \cite{MRSZ11} which
improved the bit complexity: its bit complexity is optimal and equal
to $O(\log n)$ {\it w.h.p.}  
An experimental comparison between 
\cite{Luby} and \cite{MRSZ11} is presented in \cite{BK13}.
If we remove the constraint on the size of
messages or on the anonymity, recent results have been obtained for
distributed symmetry breaking (MIS or colouring) in
\cite{KothaP11,BarenE14}.
Afek et al. \cite{Afek13}, 
 from considerations concerning the development of
certain cells, studied the MIS problem in the discrete beeping model 
 $BL$
as presented in \cite{Cornejo10}. They consider, in particular,
the wake-on-beep model (sleeping nodes wake up upon receiving a beep) and
sender-side collision detection $B_{cd}L$: they give an $O(\log^2 n)$ rounds 
MIS algorithm. 
Jeavons et al. \cite{JSX16} present in the model $B_{cd}L$ 
a randomised algorithm
with feedback mechanism  whose expected time to compute a MIS is
$O(\log n)$.  



In the model of point to point 
message passing,
node colouring is mainly studied under two assumptions: (1) nodes
have unique identifiers, and more generally, they have an initial
colouring, and (2) every node has the same initial state and initially
only knows its own edges.  If the nodes have an initial colour, Kuhn
and Wattenhofer \cite{Kuhnco} have obtained efficient time complexity
algorithms to obtain $O(\Delta)$ colours in the case that every
node can only send its own current colour to all its neighbours.  In
\cite{Johansson}, Johansson analyses a simple randomised distributed
node colouring algorithm for anonymous graphs.  He proves that this
algorithm runs in $O(\log n)$ rounds {\it w.h.p.} on graphs
of size $n.$ The size of each message is $\log n,$ thus the bit
complexity per channel of this algorithm is $O(\log^2 n).$
The authors of~\cite{MRSZ10} present an optimal bit and time complexity Las Vegas distributed algorithm for colouring any anonymous graph in $O(\log n)$
bit rounds {\it w.h.p.} Finally, a greedy colouring algorithm is proposed in~\cite{JSX16} that extends the beeping MIS algorithm in a simple way, and shares most of its analysis.

In~\cite{Cornejo10}, Cornejo and Kuhn study the interval colouring problem:
an interval colouring assigns to each node an interval (contiguous fraction)
of resources such that neighbours do not share resources
(it is a variant of node colouring). They assume that each node
knows an upper bound of the maximum degree $\Delta$
of the graph. They present in the  beeping
model $BL$ a probabilistic  algorithm which never stops
and stabilises with a $O(\Delta)$-interval colouring  in
$O(Q \log n)$ slots.

Kothapalli et al. consider the family of anonymous rings and show in
\cite{KOSS} that if only one bit can be sent along each edge in a
round (point to point message passing model), 
then every Las Vegas distributed node colouring algorithm (in
which every node has the same initial state and initially only knows
its own edges) needs $\Omega(\log n)$ rounds {\it w.h.p.} to
colour the ring of size $n$ with any finite number of colours.
Kothapalli et al.  consider also the family of oriented rings and they
prove that the bit complexity in this family is $\Omega(\sqrt{\log
  n})$ {\it w.h.p.}

The authors of~\cite{FMRZ13} present and analyse Las Vegas distributed algorithms which compute a MIS or a maximal matching for anonymous rings (in the point to 
point message passing model).  Their bit complexity
and time complexity are $O(\sqrt{\log n})$ {\it w.h.p.}


Emek and Wattenhofer introduce in \cite{EmekW13}
a model for distributed computations which resembles the beeping model:
networked finite state machines (nFSM for short). This model enables the sending
of the same message to all neighbours of a node; however it is asynchronous,
the states of nodes belong to a finite set, the degree of nodes is bounded
and the set of messages is also finite. In the nFSM model they give a
2-hop MIS algorithm for graphs of size $n$ using a set of messages of size $3$
with a time complexity equal to $O({\log^2 n}).$










\section{Conclusion}
We presented in this paper a number of design patterns which make the design and analysis of beeping algorithms simpler. They also make more apparent the connections between several algorithms which seem at first unrelated. This was illustrated through a number of algorithms and analysis. In addition, we investigated the comparative cost of various beeping models and presented a canonical method for transforming Las Vegas algorithms in stronger models into Monte Carlo algorithms in weaker models.

\section*{Acknowledgments}
We thank the anonymous referees for their comments on an earlier version of this article, which helped us improve significantly the presentation of our results.

\end{document}